\documentclass[10pt,letterpaper]{article}
\usepackage[margin=1.25in]{geometry}
\usepackage{fullpage}
\usepackage{subcaption,amsmath,amssymb,amsthm,mathrsfs,mathabx}
\usepackage{thmtools,booktabs,xspace}
\usepackage{dsfont}
\usepackage[normalem]{ulem}
\usepackage{multirow}
\usepackage{comment}
\usepackage{hyperref}
\usepackage[svgnames]{xcolor}
\usepackage[capitalise,nameinlink]{cleveref}
\definecolor{links}{RGB}{11, 85, 255}
\definecolor{cites}{RGB}{0, 200, 0}
\definecolor{urls}{RGB}{255, 116, 0}
\hypersetup{colorlinks={true},linkcolor={links},citecolor=[named]{cites},urlcolor={urls}}
\usepackage{ocgx2}
\usepackage[numbers,sort&compress]{natbib}
\usepackage{booktabs} 
\usepackage[ruled]{algorithm2e} 

\usepackage{tikz,bm,thm-restate,cleveref}

\usetikzlibrary{arrows}
\usetikzlibrary{positioning}

\newtheorem{thm}{Theorem}
\newtheorem{claim}{Claim}

\newtheorem{corr}{Corollary}
\newtheorem{definition}{Definition}
\newtheorem{remark}{Remark}

\declaretheorem[name=Theorem,sibling=thm]{restate-thm}

\newcommand{\cp}{c_{\P}}
\newcommand{\cq}{c_{\Q}}
\newcommand{\pp}{\pi_\P}
\newcommand{\pq}{\pi_\Q}
\newcommand{\lambdap}{\lambda_\P}
\newcommand{\lambdaq}{\lambda_\Q}
\newcommand{\xp}{X_{\P}}
\newcommand{\xq}{X_{\Q}}
\newcommand{\zp}{z_{\P}}
\newcommand{\zq}{z_{\Q}}
\newcommand{\ind}[1]{\mathbb{I}_{#1}}
\newcommand{\R}{\mathbb R}
\renewcommand{\L}{\mathcal L}
\newcommand{\M}{\mathcal M}
\renewcommand{\P}{\mathcal P}
\newcommand{\Q}{\mathcal Q}
\newcommand{\Ci}{\mathscr{C}}
\newcommand{\E}[1]{\mathbb E\left[ #1 \right]}

\newcommand{\F}{\mathcal F}
\renewcommand{\Pr}[1]{\mathbb{P} \left( #1 \right)}

\newcommand{\argmax}{\operatorname{arg\,max}}

\newcommand{\PBP}{Pandora's Box Problem\xspace}
\newcommand{\PBPTree}{Tree-Constrained Pandora's Box Problem\xspace}
\newcommand{\PBPDag}{Pandora's Box Problem with order constraints\xspace}

\begin{document}

\title{Pandora's Box Problem with Order Constraints}

\author{Shant~Boodaghians\\
University of Illinois\\
Urbana, IL, USA\\
\texttt{boodagh2@illinois.edu}
\and
Federico~Fusco\\
Sapienza University, DIAG\\
Rome, Italy\\
\texttt{fuscof@diag.uniroma1.it}
\and
Philip Lazos\\
Sapienza University, DIAG\\
Rome, Italy\\
\texttt{lazos@diag.uniroma1.it}
\and
Stefano~Leonardi\\
Sapienza University, DIAG\\
Rome, Italy\\
\texttt{leonardi@diag.uniroma1.it}
}

\maketitle

\begin{abstract}
The Pandora's Box Problem, originally formalized by Weitzman in 1979, models selection from a set of options each with stochastic parameters, when evaluation (i.e. sampling) is costly. This includes, for example, the problem of hiring a skilled worker, where only one hire can be made, but the evaluation of each candidate is an expensive procedure. 

Weitzman showed that the Pandora's Box Problem admits a simple and elegant solution which considers the options in decreasing order of the value it which opening has exactly zero marginal revenue.
We study for the first time this problem when the order in which the boxes are opened is constrained,
which forces the threshold values to account for both the depth of search, as opening a box gives access to more boxes, and breadth, as there are many directions to explore in.
Despite these difficulties, we show that greedy optimal strategies exist and can be efficiently computed for {\em tree-like} order constraints.

We also prove that finding approximately optimal adaptive search strategies is NP-hard when certain matroid constraints
are used to further restrict the set of boxes which may be opened,
or when the order constraints are given as reachability constraints on a DAG.
We complement the above result by giving approximate adaptive search strategies based on a connection between optimal adaptive strategies and non-adaptive strategies with bounded adaptivity gap  for a carefully relaxed version of the problem.
\end{abstract}

\section{Introduction}

Stochastic search is an important problem in many fields of application, and has seen theoretical study. 
Modelling evaluation as costly captures the commonly observed trade-off between exploration and exploitation, and is also a natural assumption in many settings. 
Furthermore, it is common to only select one of the many evaluated alternatives, such as when searching for a skilled employee, or choosing a house. 
These ideas were formalized as the {\em \PBP} by Weitzman, in 1979~\cite{weitzman}.  
In this setting, we are presented with various alternatives modeled by a set of boxes  $B=\{b_1,\,\dotsc,\,b_n\}$, 
where box $b_i$  costs $c_i$ to open, and has random  payoff $X_i$, whose distribution is known.  A strategy $\pi$ is  a rule which determines adaptively whether to terminate the search, and if not, which box to open next.  
The goal is to choose a strategy $\pi$ which maximizes, in expectation, the following objective:
            \[\textstyle
        	\max_{i\in S(\pi)} X_i - \sum_{i\in S(\pi)} c_i,
            \]
where $S(\pi)$ is the set of boxes opened by strategy $\pi$.
Only \emph{one} reward can be kept in the end, but we must pay for all opened boxes.

Despite the broad range of search strategies available, Weitzman showed that the solution to this problem boils down to a simple but subtle strategy.
Each box is assigned a \emph{reservation value} $\zeta_i$ satisfying the equation $\E{\max\{0,(X_i-\zeta_i)\}-c_i}=0$. This is precisely the value that the player would need to have collected, in order to be indifferent between proceeding or not, and is used as a proxy for the value of the box.
The optimal strategy is to greedily open the boxes in descending order of reservation value, and to stop when there is no box left or when the maximum reward seen in the past is greater than the reservation value of the next box. Note that the order of exploration is not adaptive, but the stopping time is.

In this paper, we enrich the classical \PBP by adding restrictions on the order in which boxes can be accessed.
To motivate this, consider the task faced by a funding agency or a research and development department. There may be different high-level directions to explore, with differing endpoints.
Is it better to follow a longer line of research with a higher probability of achieving moderately interesting results, or take a riskier but shorter path which could contain a miraculous discovery at the end? In both cases, the final goal as well as the way to reach it will greatly affect the costs and payoffs generated.

With these extra constraints, the search for the optimal strategy becomes much more complex. We seek to determine the order in which to assess the options and when to stop searching, given prior knowledge of the values, the structure of the constraints and the realizations of the observed rewards. 
Weitzman's classical greedy strategy~\cite{weitzman} does not apply:  opening an expensive box with little reward may be needed in order to access more valuable and inexpensive boxes. 

As a motivating case for this line of work, we show that when the order constraints are tree-like --- that is when the boxes are nodes of a tree, and can only be opened after their parents are --- then optimal search procedures may be found. 
This also applies for forests.
Furthermore, we investigate the relation between such order constraints and other models of the \PBP which have been studied in the past, such as in \cite{singla2018price}. We formalize these results in the following section.

\subsection{Overview of Results and Methods}\label{sec:overview}
In this paper we focus on the \PBP with constraints on the order of exploration.
In what follows, {\em polynomial time} means polynomial in the number of boxes and the size of the support of the random variables. This implicitly suggests that the random variables are discrete and have support size polynomial in the number of boxes. 
This assumption is not restrictive: we extend the techniques of \citet{guo2019generalizing} to prove that a polynomial number of samples is enough to get an $\epsilon$ additive approximate of the solution for the \PBPDag. Moreover, these bounds are tight
for the \PBPTree. 
To approximate the performance of optimal policies within an additive $\epsilon>0$, with probability $1-\delta$, for $n$ boxes with  rewards and costs supported in [0,1], it suffices to learn each reward distribution from $\tilde O(n^3/\epsilon^3)$ samples. 
For the special case of tree constraints, it suffices to take $\tilde O(\tfrac n{\epsilon^2})$ samples.
We have hidden $\operatorname{poly}(\log(\tfrac n{\epsilon\delta}))$ terms.
See Appendix~\ref{app:learning} for details. 

These ideas allow us to proceed with the analysis while assuming without loss of generality that the random variables have support $poly(n)$.
We further note that this result suggests the methods in this paper are robust to imprecise knowledge of the distribution of the rewards on the boxes.

\subsubsection{The \PBPTree.}
We first consider the Tree case, where  we present an optimal strategy with a nice structure analogous to Weitzman's. 
\begin{thm} When order constraints are given by a rooted tree over the boxes, there exists an optimal-in-expectation strategy 
of the following form:
first, label each box with a ``threshold'', an order-aware analog of Weitzman's reservation value. Then:
\begin{itemize}
    \item From the boxes that can be opened next, choose the one with the largest threshold.
    \item Terminate if the max observed value exceeds this threshold, otherwise open the box and repeat.
\end{itemize}
Furthermore, the optimal thresholds can be computed in polynomial time.
\end{thm}
Notice that such threshold strategies are simple and intuitive, 
and enforce the desirable property that the order of the exploration is fixed up to tie-breaking. 
Furthermore, that an optimal strategy has this form is surprising, as this need not hold for general constraints, as discussed below.
The definition of these thresholds has to address the following additional complications:

        \begin{itemize}
            \item[(i)] \textbf{Depth:} The value of a box is not only given by its reward and cost, but also from the possibilities its opening makes accessible. This effect propagates level after level: even the deepest of the leaves can influence the decision to open the root.
            \item[(ii)] \textbf{Breadth:} A key feature is the order in which the boxes are opened. It is difficult to model and optimize the interplay between different explored branches of a tree, as distant directions of exploration must  be compared at every time step.  
        \end{itemize}
To overcome these difficulties, the first step is to generalize the reservation values used by Weitzman to the setting where the boxes are constrained to be opened in a fixed order $[b_1,\, b_2,\,\dotsc,\,b_n]$.
These values must take into account the \emph{future} as well as the present.
Our solution consists in defining the threshold value of the generic box $i$ according to a random  \emph{stopping time} $\tau^*(y,i)$, which indicates the last box that will be opened {\em playing optimally} given that the player has already found reward $y$ and is in front of box $i$. We call this threshold value $z_i$. 
It relates to the original definition as follows: 
\[\E{\left(X_i-\zeta_i\right)_+-c_i}=0
\quad
\longleftrightarrow
\quad
\E{\left( \textstyle\max_{j=i}^{\tau^*(z_i,i)}X_j-z_i\right)_+-\textstyle\sum_{j=i}^{\tau^*(z_i,i)}c_j}=0,\]
where $(a)_+:=\max\{a,0\}$.
From these stopping times and threshold values we can infer that certain runs of boxes are essentially treated as one big box (which we refer to as a \emph{macrobox}). If the algorithm decides to enter it, the exploration will either finish inside this macrobox or a decision will be made to enter the next one after having exhausted it. Moreover, these threshold values can be computed in polynomial time by a dynamic programming procedure.

For a single line, this reasoning may seem straightforward, but this property still holds when the constraints consist of a union of disjoint parallel lines. A na\"ive dynamic program would not be effective, as the state space is exponential in the number of lines.
However, the threshold strategy which uses the reservation values computed for each line independently is still optimal: the algorithm will always enter the best available macrobox and either terminate search inside of it, or move on to another one, possibly from a different line. 
Surprisingly, the same approach works for \emph{trees and forests}. 
Proceeding from leaves to the root it is possible to linearize the trees and use the definition of reservation value to induce a threshold strategy which is indeed optimal. 

\subsubsection{Impossibility and Hardness Results.}
Unfortunately, these results do not extend to slightly more general constraint structures. We show below that it is NP-hard to approximate an optimal solution to the problem with some types of matroid constraints or more general order constraints.
Remark that the standard notion of approximation --- {\em i.e.} finding a solution whose performance is at least a multiplicative factor of the optimal solution in expectation --- is not meaningful in this setting: a hard example ({\em e.g.} the hardness proof of \Cref{sec:hardness}) can be modified by adding a large-cost-no-payoff dummy box at the root of the tree. For the right cost, optimal strategies would have positive revenue, but approximately optimal strategies would have negative revenue.

We define below a modified notion of approximately optimal solution, which avoids these concerns. 
Note that similar metrics have been used before, in works such as \cite{feldman2019guess,harshaw2019submodular,sviridenko2017optimal}.  

\begin{definition}[Approximately Optimal Solutions]\label{def:appx}
     In this paper, we consider a notion which approximates only the reward term while paying similar costs.
    Formally, we seek a strategy $\hat \pi$ such that for any other rule $\pi$,
    \[
        \textstyle\E{\max_{i\in S(\hat \pi)} X_i - \sum_{i\in S(\hat \pi)} c_i}
        \geq 
        C\cdot \E{\max_{i\in S(\pi)} X_i} - \E{\sum_{i\in S(\pi)} c_i}
    \]
    for some universal constant factor $C \in (0,1]$, where $S(\pi)$ is the (random) set opened by strategy $\pi$.
\end{definition}

\begin{thm}
Consider the \PBPDag when either (i) a matroid constraint is added to the tree constraint or (ii) the tree constraint is generalized to a DAG. 
In either case, it is NP-hard to find a 0.9997-approximately optimal solution, and furthermore, the optimal solution need not have a fixed order of exploration.
\end{thm}
This approximation hardness is shown in the sense of Definition \ref{def:appx}, but extends to the classical notion, since it is stronger.
We finally remark that proving hardness of approximation for stochastic problems needs to address the effects of the randomness on the objective and the search trajectory, and as such is challenging. Moreover, the structure of the $\max_S X_i - \sum_S c_i$ objective makes optimal {\em random} solutions difficult to ``hide'', in a standard combinatorial sense.

\subsubsection{Approximation Results.}
The hardness result above, along with recent work on modified versions of the \PBP~ \cite{beyhaghi2019pandora,singla2018price,singla2018combinatorial,kleinberg2018delegated,esfandiari2019online}, motivates the study of approximation algorithms for the more general case of order constraints. Of the above citations, the ones closest to our setting are \cite{singla2018price,singla2018combinatorial}, where the author reduces the \PBP in the presence of downwards-closed constraints to adaptive maximization of non-negative submodular functions. 
The key concept in their work is the {\em adaptivity gap}, i.e. the ratio between the best adaptive solution and the best non-adaptive one for this new problem. 
We similarly show the following.

\begin{restatable}{restate-thm}{adaptivity}
\label{thm:adaptivity}
Consider the \PBP with constraints modeled by some prefix closed family --- a generalization of order constraints, defined in Section~\ref{sec:model},
then for every adaptive strategy $\pi$, there exists a non-adaptive strategy, i.e. a feasible set $S$, such that the following holds\footnote{
        To further justify the benchmark of Definition~\ref{def:appx}, we show in Appendix \ref{app:counterexamples} that approximation in the traditional sense with a non-adaptive set is impossible, as there exist examples where this classical adaptivity gap is arbitrarily large.}:
    \begin{equation}\textstyle
    \E{\max_{i \in S} X_i - \sum_{i \in S}c_i} \geq \tfrac{1}{2} \E{\max_{i\in S(\pi)} X_i} - \E{\sum_{i\in S(\pi)} c_i}
    \end{equation}
\end{restatable}

This result effectively reduces the problem of approximating adaptive strategies to the problem of selecting optimal non-adaptive sets. 
It should be noted, however, that in full generality of matroid constraints or precedence constraints, this could be intractable. We therefore follow an alternative approach. We show that there exists an {\em adaptive} strategy whose performance is better than that of every non-adaptive set, simultaneously.
\begin{thm} For any tree-constraint, and any further downwards-closed constraint on the set of boxes that can be opened, there exists an adaptive strategy $\hat\pi$ such that for any fixed set $S$, $\hat\pi$ performs better in expectation than non-adaptively opening $S$. Note that $\hat \pi$ does not depend on $S$.
Furthermore, when the downwards-closed constraints are given by generalized knapsack constraints, or any ``sufficiently oblivious'' matroid constraint (as defined in Section~\ref{sec:appx-algs}), the strategy $\hat \pi$ can be computed efficiently.
\end{thm}

\subsection{Related Work}
    As discussed above, the starting point for this theory is the 1979 Weitzman's paper~\cite{weitzman},which was at the time a generalization of preceding results in special cases, namely \citet{kadane1977optimal,kadane1969quiz}.
    In the following years it was highlighted the similarity with the multi-armed bandits problem, which uses the notion of Gittins index as a similar cutoff value (e.g. \citet{weber1992gittins,frostig1999four}). Indeed the reservation value of the classical \PBP is a version of Gittins index, \citet{dumitriu2003playing}.
    \citet{weiss1988branching}, and in particular \citet{keller2003branching}, deal with a similar problem, the branching bandit process.
    The branching process resembles the \PBPTree, though in that model, the process does not terminate, and the revenue is measured as an infinite-horizon discounted sum of payoffs, even if some finite horizon results are showed. The main difference with our works is that we focus on maximizing the largest reward minus the exploration costs, an objective function they cannot capture with their techniques, moreover their solution is defined with an implicit formula that becomes rapidly cumbersome as a function of the height of the tree. Interestingly enough they prove the optimality of a threshold strategy which is, in spirit, quite similar to ours.
    
    \citet{kleinberg2016descending}, borrowing from the language of finance, introduce the {\em covered call value} of a box, which is the minimum of the reservation value, and its true (random) value. They show that the expected performance of any search strategy is at most the expected covered call value of the last kept box. This inequality is tight for any method which immediately terminates search when the value of an opened box is greater than its reservation value.
    This novel point of view on \PBP started a new interest in the problem. \citet{olszewski2015more} investigate the existence of moving threshold strategies to address more general objective functions in the Pandora's setting, while \citet{doval2018whether} and \citet{beyhaghi2019pandora} analyze a setting in which a box can be chosen without paying its cost while retaining its expected reward.
    We highlight that these techniques do not extend to order-constrained settings.
    
    \citet{singla2018price,singla2018combinatorial} exploits the notion of {\em surrogate} box, an analog of the {\em covered call value}, to reduce the \PBP in the presence of downwards-closed constraints to adaptive maximization of nonnegative submodular functions, and bound the {\em adaptivity gap} of this problem, {\em i.e.} the ratio between the best adaptive solution and the best oblivious one as in \citet{gupta2017adaptivity} and later \citet{adaptivityGapProof}. Recently \citet{kleinberg2018delegated} and \citet{esfandiari2019online} studied a connection between the \PBP and another well known optimal stopping problem, termed the Prophet Inequalities.


\section{Model and Preliminaries}
\label{sec:model}
        In this section, we formally present our model, and give preliminaries.
        Recall, as the player, we adaptively open a constraint-satisfying set of boxes, paying for each one opened, and learn the (random) value of each box only after paying the cost of opening. The final payoff received is the largest value observed. This is formalized below.
    
        \paragraph{The \PBPDag.}
            We are given a set of boxes $B=\{b_1,\,\dotsc,\,b_n\}$, where $b_i$  costs $c_i$ to open, and has random payoff $X_i$, whose distribution is known. The $\{X_i\}_{i=1}^n$ are independent and need not be identically distributed. 
            
            A strategy $\pi$ is a rule which determines, at any integer time $t\geq 0$,  whether to terminate the search and, if not, which box to open next. 
            The strategy may depend on the time $t$, the values observed in the past, the structure of the problem and some extra randomness. We use equivalently the terms strategy, rule and policy.
            
            Let $S_t(\pi)$ denotes the (random) set of boxes that have been opened before time $t$ (included) by strategy $\pi$, and let $\tau_{\pi}$ be the stopping time given by the same strategy. We use the shortcut $S(\pi)$ to denote $S_{\tau_{\pi}}(\pi)$, the final set of opened boxes following strategy $\pi$.
            Given constraint-set $\mathcal F\subseteq 2^B$, $\pi$ is said to be $\mathcal F$-feasible if $S_t(\pi)\in \mathcal F$ with probability 1, for all $t$. 
            Our goal is to choose a policy $\pi^*$ which maximizes, in expectation, the following objective:
            \[
    \textstyle\E{        	\max_{i\in S(\pi^*)} X_i - \sum_{i\in S(\pi^*)} c_i}
            \]
            Such strategies are called optimal.
            
        \paragraph{Threshold Strategies.}
            A rule $\pi$ is said to be a threshold strategy if it pre-computes a collection of {\em threshold values}, and greedily opens the boxes following these values, stopping when the amount earned is greater than the threshold of all remaining legal moves.
            Formally, the strategy is defined by a {\em threshold function} $z:B\to \R$ and works as in Algorithm \ref{alg:threshold} below.

            \begin{algorithm}
            \caption{Threshold strategy}
                \KwData{Distributions of the random rewards, box costs and a threshold function $z:B\to \R$}
                $S_0 \gets \emptyset$, $y \gets 0$,  $t \gets 0$\\
                \While{$y<\max \{z(b) | b \in B \backslash S_t \text{ and } \{b\} \cup S_t \in \F \}$}{
                Let $\hat b \in \argmax\{z(b) | b \in B \backslash S_t \text{ and } \{b\} \cup S_t \in \F \}$, tie-breaking arbitrarily\\
                Open box $\hat b$, observe reward $\hat X$ and pay cost $\hat c$\\
                $\ S_{t+1} \leftarrow S_t \cup \{\hat b\}, \ y \leftarrow \max\{y,\hat X\}, \ t \leftarrow t+1$\\
                }
            \label{alg:threshold}
            \end{algorithm}
            
            Observe that, given some consistent tie-breaking rule, the order of exploration is fixed, as the next box to consider only depends on the reservation values.
            We remark that any threshold function is defined \textit{a priori}, $i.e.$ it does not depend on the observed rewards, but only on the costs, the random distributions, and $\F$. 
            The following facts follow by definition:
            \begin{claim}
                Let $\pi$ be a threshold 
                strategy,
                where ties in the thresholds are solved arbitrarily but consistently.
                Then the following hold true:
                \begin{enumerate}
                    \item \textbf{Fixed order:} Following $\pi$, \ $\Pr{b_i \text{ is opened before } b_j}>0 \implies \Pr{b_j \text{ is opened before } b_i}=~0$, for all $i\neq j$.
                    \item \textbf{Efficiency:} If the threshold function $z$ is efficiently computable, then so is $\pi$.
                \end{enumerate}
            \end{claim}
        \paragraph{Order Constraints}
            In this paper, we focus on {\em order constraints}, where some boxes are required to be opened after others.
            Such constraints can be represented by a rooted tree $T$ called precedence tree, whose nodes are the boxes. A box may only be opened if its parent has already been explored. Without loss of generality we can assume that $T$ is connected, {\em i.e.} it is actually a tree, and that there is a unique node $r \in B$, which is the starting box --- if not, it suffices to add a dummy root box with no cost and no reward. 
            Formally, then, given the precedence tree $T$, the feasible sets $\mathcal F_T$ are:
            \[
        	    \mathcal F_T:= \{ S\subseteq B \,|\, \forall u\in S \setminus \{r\},\ \textsc{parent}(u) \in S \}
            \]
        
        \paragraph{Prefix closed Constraints} Order constraints may be seen as a special case of the more general {\em prefix-closed} constraints, which simply assert that for any legal sequence of moves, any truncation of this sequence is also legal. Formally, given a set of boxes $B$ and a set $\Ci$ of possible orders of exploration, we say that $\Ci$ is prefix closed if for every $C \in \Ci$, every prefix of $C$ is also in $\Ci$. 
        Note that the order constraints defined above are a special case of this. 
        Furthermore, intersecting any combination of order and downwards closed constraints results in some prefix closed family. 

        \paragraph{Distributional Assumptions} 
            As discussed in the beginning of Section~\ref{sec:overview}, it is not restrictive to assume that the random variables $\{X_i\}_{i = 1,\dots,n}$ are discrete, and are supported on $s<\infty$ values. When we say an algorithm runs in polynomial time, we mean polynomial in~$s$ and~$n$. 

        \paragraph{Preliminaries}
            Before proceeding with the details of the analysis, we would like to note that the \PBPDag admits a na\"ive, albeit exponential-time solution: it suffices to solve a dynamic program whose states are all pairs $(S,y)$ where $S\subseteq B$ is a set of boxes, and $y\in \R$ is the max value observed. 
            
            The literature on Markov Decision Processes (e.g. \cite{puterman1994markov}) allows us to fix basic properties of optimal strategies: there exists an optimal strategy $\pi^*$  which is a Markovian policy mapping states to actions, {\em i.e.} the optimal next box to open is deterministic function of the state $(S,y)$. 
    
        \paragraph{Notation} 
            In the following we use interchangeably $\max(a,b)$ or $a \vee b$ to denote the largest between two reals $a$ and $b$. For the smallest we use $\min(a,b)$ or $a \wedge b$. As a further simplification $(a)_+ := \max\{a,0\}$. We use the following simple equality repeatedly: $a \vee b - b = (a-b)_+.$


    \section{Optimal Search on Tree Constraints}
    
        In this section, we present and analyze an optimal-in-expectation search procedure for the \PBPTree.
        The classical result of Weitzman~\cite{weitzman} shows that the unconstrained problem may be solved by a simple threshold strategy as discussed above. Formally, 
        
        \begin{definition}[Reservation Value and Pandora's Rule]
        \label{def:reservation}
            Given a box $b$ with cost $c>0$ and nonnegative random reward $X$, we define the reservation value $\zeta$ of $b$ as the smallest solution to
            $
                c=\E{\left( X-\zeta\right)_+}.
            $
            It can be shown that if $X$ has finite mean, then the reservation value is well defined.
            The threshold strategy using the reservation values as thresholds is termed {\em Pandora's Rule}.
        \end{definition}
        The power of this strategy is that the reservation value depends only on the single box, allowing us to consider each box separately, leaving the problem dramatically more tractable. 
        In Sections~\ref{sec:one-line-opt} and~\ref{sec:multi-line-opt}, we present a solution to the \PBPDag given by one directed line, and a collection of disjoint, directed lines, respectively.
        In Section~\ref{sec:tree-constraints-opt}, we show that, in fact, solving the problem for generalized, rooted tree constraints follows immediately as a corollary, and as such, the heart of the technical contribution lies in \Cref{sec:multi-line-opt}. 
        
        \subsection{\PBP on a single Line}\label{sec:one-line-opt}
        We introduce here the simplest order constraint, given by an ordered path. 
        This constraint may seem trivial at first, since the order of exploration is fixed, and it suffices to determine the stopping time, but it illustrates the main difficulty of order constraints: the intrinsic value of a box is not given \textit{only} by its cost and its random reward, but also by the other boxes that are made available after its opening. Consider, for example, the line consisting of one box with cost but no reward, followed by another with reward but no cost. In this sense, a na\"ive threshold strategy in the sense of Weitzman's result does not immediately suffice.
        
        As mentioned, solving the problem on the single line is not of great consequence. However, the concepts introduced in its solution are very informative for the following, so we present them separately for clarity of exposition.

        Let $\L=[b_1, b_2, b_3,$ $ ..., b_n]$ be the ordered set of boxes, where for all $i$, $b_i$ costs $c_i$ to open, and gives random reward $X_i$, as usual. 
        As an additional constraint we have that the boxes must be opened in order of their indices: 
        box $b_i$ can be opened only if it is the root or box $b_{i-1}$ has already been opened. 
        After having opened the first $i$ boxes and collected some reward, it is a simple exercise of Dynamic Programming to determine whether to open the $(i+1)$-th box and proceed optimally, or to terminate.
        The DP had $O(ns)$ states for $n$ boxes with reward supported on $s$ values.
        
        However, the goal of this section is to illustrate structural properties which we will use in the general tree setting.
        It is clear that the set of possible strategies in this simple setting coincides with the set of all (random) stopping times $\tau$, with respect to the filtration given by $X_1, X_2, \dots, X_n$. 
        For any $x \in \mathbb{R}_+$, $\tau$ and $i=0,1,\dots,n,n+1$ we denote $(x,i)$ as the state in which box $b_i$ is the next accessible box and $x$ is the largest reward uncovered so far, and $\tau(x,i)$ as the (random) stopping time conditioned on being in that state, i.e. conditioned on the events $\tau \geq i-1$ and $\max_{j=1}^{i-1}X_j = x$. 
        
        Before, and in what follows, we assume sum and the $\max$ operations over empty sets to have value $0$.
        We can define the expected future reward following $\tau$, starting in state $(x,i)$, as:
        \[
        \phi^{\tau}(x,i)\ :=\ \E{\max\left\{x,\ \textstyle\max_{j=i}^{\tau(x,i)}X_j\right\}-\textstyle\sum_{j=i}^{\tau(x,i)} c_j},
        \]        where, by convention, we set $\phi^\tau(x,i)$ to be 0 when $\tau$ is ill-defined,
        {\em i.e.} the event $\tau(x,i)$ is conditioning on occurs with probability 0.
        In addition to that we define $\Phi(x,i) = \max_{\tau} \phi^{\tau}(x,i).$
        
        The objective is to choose a $\tau^*$ which maximizes $\phi^{\tau}(0,0)$, corresponding to $\Phi(0,0)$. 
        By the remarks in the Preliminaries, we restrict ourselves to stopping times which are deterministic functions of the $X_i$'s, i.e. the decision to stop depends deterministically only on the structure of the problem and the realizations of the rewards experienced so far.
        
         \begin{remark}
        It is natural to ask whether a fully-deterministic, i.e. non adaptive, stopping time is a valid strategy. It turns out that such an approach may be arbitrarily worse than an adaptive stopping time.
        We provide an example in Appendix~\ref{app:counterexamples}, which resembles ones found in~\cite{singla2018combinatorial}, and the extended version of~\cite{singla2018price}.  

    \end{remark}
        
        The notion of conditional stopping time defined above allows us to give the following definition, which is analogous to Definition \ref{def:reservation}:
        \begin{definition}
        \label{def:reservationLine}
            Let $\L=[b_1,b_2\dots,b_n]$ be a line of $n $ boxes, then for every $i=1, \dots, n$ we can define the generalized reservation value of box $b_i$, denoted $z_i$, as the smallest solution to
            \begin{equation}
            \label{eq:reservationLine}
                \E{\left(\textstyle \max_{j=i}^{\tau^*(z_i,i)}X_j-z_i\right)_+-\textstyle\sum_{j=i}^{\tau^*(z_i,i)}c_j}=0
            \end{equation}
        where $\tau^*(x,k)$ is an optimal random stopping time given that the largest reward sampled in the past has been $x$ and the player has just opened box $b_{k-1}$, or nothing if $k = 1$.
        \end{definition}
        Whereas the $\zeta$ value in Definition~\ref{def:reservation} was effectively the value collected in the past for which we were indifferent between opening a box or not, the $z_i$ value in this definition is the past collected value for which we are indifference between proceeding (optimally) along the line of boxes or not. 
        The following claim ensures that this is well-defined and its proof can be found in Appendix \ref{app:proofs}. \\[-2ex]
        \begin{claim}
        \label{cl:reservazionLine}
            Definition \ref{def:reservationLine} is well posed, in that the smallest solution of $(\ref{eq:reservationLine})$ exists and does not depend on the choice of $\tau^*$. Also, if $z_i>0$, some optimal stopping time $\tau^*(z_i,i)$ does not stop at $i-1$. 
            Finally, $z_i$ is  the value for which we are indifferent between stopping and proceeding optimally.
        \end{claim}
        
        We highlight here that, as $z_i$ is defined, the optimal strategy $\tau^*(z_i,i)$ is ambiguous: we are indifferent between stopping ($\tau^*=i-1$), and proceeding ($\tau^*(z_i,i)=\tau^*(z_i\vee X_i,i+1)$). 
        By convention we will always refer to the latter.
        A brief calculation confirms that this new definition contains as a special case the classical reservation value, viewing single boxes as length-1 lines.
        We will henceforth omit the term ``generalized'' in reference to reservation values, when clear from context.
        We call the threshold strategy associated to these reservation values Generalized Pandora's Rule.  
        
        \begin{thm}\label{thm:general-pandora-opt-one-line}
             The Generalized Pandora's Rule for the Line is optimal and can be computed in polynomial time and space.
        \end{thm}
        \begin{proof}
            It follows by definition, and by Claim~\ref{cl:reservazionLine}, that it is in our interest to proceed if the largest value seen is less than $z_i$, and to stop, if the largest value seen is greater. Thus, any deviation from this threshold strategy is sub-optimal. 
            Furthermore, $\Phi(x,i)$ may be computed by a simple dynamic program solved in decreasing order of $i$. The reservation price of box $i$ is the smallest point in the column for $i$ where $\Phi(x,i)=x$. 
        \end{proof}
        
        We can define the procedure $\textsc{computeThreshold}(b,\L)$ as outputting the exact reservation price of $b$ if it were added as a prefix to $\L$. 
        This can be done in polynomial space and time by simply referring to the dynamic programming table for $\Phi$ and following Definition \ref{def:reservationLine}.
        This simple function plays a crucial role in the design of the final algorithm for the Tree constrained case. 
        
        We conclude this section by giving some properties of the reservation values:
        \begin{claim}
        \label{prop01}
            Given a line $\L=[b_1, \dots, b_n ]$ the following statements hold true for every $i=1, \dots,n$:
            \begin{enumerate}
                \item $z_i$ can only increase if something is added at the end of the line $\L$;
                \item For every $i$ let $d(i)$ be the $\min\{t \geq i|z_{t+1}<z_i \text{ or t = n}\}$, then $z_i$ depends only on the prefix $[b_i, \dots, b_{d(i)}]$. If $d(i)=i$, then $z_i$ depends only on $b_i,$ i.e. $z_i = \zeta_i$.
                \item If $ z_i < z_j$ for all $j>i$, then the optimal stopping time $\tau^*(z_i,i)$ given by the Generalized Pandora's Rule does not depend on $z_i$.
                In particular $\tau^*(z_i,i)=\tau^*(y,i)$ for all $y \in [0,z_i].$ 
            \end{enumerate}
        \end{claim}
        \begin{proof}
            The first property derives from the fact that $z_i$ is a fixed point of $\Phi(\cdot,i)$, and hence if something is added at the end of the line the expected revenue can only increase. For the other two claims, it is sufficient to observe that if an optimal play starts with some $y$, then $\tau^*(y,i)<j$ for all $j>i$ such that the reservation value $z_j<y.$
        \end{proof}
 
    \subsection{\PBP on a Union of Lines}
    \label{sec:multi-line-opt}
        We wish to generalize to a union of disjoint lines. 
        Observe that this setting contains both the classical \PBP, and the single-line-constrained case as sub-problems. In essence, it captures both the breadth of exploration from the unconstrained setting, and the dependence on depth of the intrinsic value of a box of the single line case.
        
            Formally, we have $k$ paths $\L_1,\,\dotsc\L_k$, where path $\L_i$ consists of boxes $b_1^i,\,\dotsc,\,b_{n_i}^i$ which can only be opened in increasing order of subscript. 
            Box $b_j^i$ costs $c_j^i$ to open, (generalized) reservation value $z_j^i$, and gives random reward $X_j^i$, following a known distribution. 
            The reservation values are defined as in the previous section, taking only the line that the box belongs to as context, $i.e.$ for any box $b_j^i$ in $\L_i$, the value $z_j^i$ is computed as if only $\L_i$ existed. 
            Furthermore, by Claim \ref{prop01}, $z_j^i$ depends only on $[b_{j}^i, \dots, b_{d^i(j)}^i] \subseteq \L_i$, where $d^i(j)$ is the equivalent of $d(j)$ for line $i$.
            Since one must consider the interactions between the different lines, the dimension of the na\"ive dynamic program is exponential in the number of paths, so we need to be more clever. 
            Surprisingly, we prove that the optimal strategy is still a threshold strategy, and the $z_j^i$'s are exactly the correct thresholds.
            The rest of this section is a proof of this fact.
            The main effort in the proof is indeed to find a way to \textit{decouple} the first box of a line from the rest 
            of that line.
            
            \begin{thm}\label{main-thm}
                The Generalized Pandora's Rule is optimal for the \PBP on Union of Lines and can be implemented in polynomial time and space.
            \end{thm}
            \begin{proof}
                As the reservation values are identical to the previous subsection, algorithmic results follow immediately.
                It remains to show optimality, which we will do by induction on the number of boxes yet to open.
                
                If there is only one box remaining, then this is a special case of the unconstrained \PBP and we know that it is optimal to follow the Generalized Pandora's Rule since it coincides with the Pandora's Rule.
                For the induction step, without loss of generality, we may re-label the sequences such that the first box in every line is labelled 1, and $z_1^1\geq z_1^2\geq \dotsm \geq z_1^k$.
                There are three actions to consider: 
                (1)~stopping, (2)~opening $b_1^1$ first, and (3)~opening $b_1^i$ first for some $i$ such that $z_1^1>z_1^i$. 
                We remark that the decision will depend on both the boxes to open and the largest value seen in the past, denoted $y$.
                
                \medskip
                 We begin by showing that stopping is optimal if and only if $y \geq z_1^1$. If $y<z_1^1$, we know that even the suboptimal strategy of opening box $b_1^1$ and playing only on $\L_1$ ignoring other lines is better than stopping. 
                If $y\geq z_1^1$ and we open $b_1^i$ for any $i$, then by induction, the optimal strategy is to go on exploring $\L_i$ without the possibility of changing line; as the $z_1^j$'s are too small.
                This again contradicts the $k=1$ case.
                \medskip
                
                Suppose, then, that $y< z_1^1$, and that we decide to open $b_1^i$, for $i>1$. 
                \smallskip
                
                \noindent\textbf{Case I:} $z_1^1\geq \dotsm\geq z_1^i> y$. 
                By induction, after opening $b_1^i$, the optimal strategy is to continue the exploration along $\L_i$ until the reservation value $z_j^i$ becomes less than $z_1^1$ (or the reward exceeds the next reservation value), then to go along $\L_1$.
                Let $\P$ be the prefix of $\L_1$ with reservation values greater than $z_1^1$. By~\Cref{prop01}, we note that $z_1^1$ depends only on $\P$. 
                Let $\Q$ be the prefix of $\L_i$ including $b_1^i$, and extended to contain all reservation values greater than $z_1^1$, and let $\Q'$ be the prefix of $\L_i\setminus \Q$ containing all $z$ values greater than $z_1^i$. Again, \Cref{prop01} implies $z_1^i$ depends only on $\Q \cup \Q'$.
                Define $z_\Q$ to be the reservation value of box $b_1^i$ considering only prefix $\Q$. By \Cref{prop01}, $z_\Q\leq z_1^i$.
                For simplicity, denote also $z_{\P}=z^1_1$ and the optimal stopping times in the two prefixes $\P$ and $\Q$ as $\tau_{\P}$ and $\tau_{\Q}$. 
                The heart of the proof is that we can treat the two prefixes $\P$ and $\Q$ as single {\em macro-boxes}, with random costs and random rewards. 
                Let $z_{S}$ be then the largest between $\zq$ and all the reservation values of the boxes accessible after having exhausted $\P$ and $\Q$.
                
                \smallskip
                \noindent Consider the following (suboptimal) executions of the algorithm: \\
                -- \textsc{Strategy A} begins by exploring $\P$, and if at the end of the exploration the largest reward is greater than $z_S$, then stops, otherwise explores $\Q$ as if no reward was found while exploring $\P$.\\
                -- \textsc{Strategy B} begins by exploring $\Q$, then plays optimally, that is it explores $\P$.
                
                Notice that B is the optimal execution of the algorithm, by induction, under the assumption that the first box opened is $b_1^i$. 
                This allows for a more compact representation: for any initial value $w$, we can define the following random variables: 
                $$
                X_{\P}(w):= \max_{\ell = 1}^{\tau_{\P}(w)} X_\ell^{\P}, \quad
                c_{\P}(w):= \sum_{\ell=1}^{\tau_{\P}(w)}c^{\P}_\ell, \quad
                X_{\Q}(w):= \max_{\ell = 1}^{\tau_{\Q}(w)} X_\ell^{\Q}, \quad
                c_{\Q}(w):= \sum_{\ell=1}^{\tau_{\Q}(w)}c^{\Q}_\ell\ ,
                $$           
                where the superscripts $\P$ and $\Q$ for the rewards and the costs simply specify to which prefix the boxes belong to.
                With this notation the expected revenue following strategy $A$ is:
                \begin{align*}
                    \E{-c_{\P}(y) + X_{\P}(y) \mathbb{I}_{X_{\P}(y) \geq z_S} + \mathbb{I}_{X_{\P}(y) < z_S}(- c_{\Q}(y)+\mathbb{I}_{X_{\Q}(y) \geq z_S} X_{\Q}(y))} + \qquad \\ + \E{\mathbb{I}_{X_{\P}(y) < z_S}\mathbb{I}_{X_{\Q}(y) < z_S}\Phi_S(X_{\Q}(y) \vee X_{\P}(y) \vee y)},
                \end{align*}
                where $\Phi_S(\cdot)$ here is the optimal expected revenue after exhausting both $\P$ and $\Q$.
                
                For strategy B, it is convenient to note that the prefix $\P$ is opened only if in prefix $\Q$ no reward is greater than $z_{\P}$. 
                This means that the stopping times $\tau_{\P}(y)$ and $\tau_{\P}(y \vee X_{\Q}(y))$ are the same in this case. 
                Hence, $X_{\P}(y)=X_{\P}(X_{\Q} \vee y)$ and $c_{\P}(y)=c_{\P}(X_{\Q} \vee y)$.
                This is the main step: for both the strategies, if both the prefixes are activated, then the stopping time of the second one is independent to the realizations in the previous prefix.
                Now the reward due to B is
                \begin{align*}
                    &\E{-c_{\Q}(y) + X_{\Q}(y) \mathbb{I}_{X_{\Q}(y) \geq z_{\P}} +\mathbb{I}_{X_{\Q}(y) < z_{\P}}\mathbb{I}_{X_{\P}(y) \geq z_S}(X_{\Q}(y) \vee X_{\P}(y))} + \qquad\\ &+ \E{-c_{\P}(y) \mathbb{I}_{X_{\Q}(y) < z_{\P}}+\mathbb{I}_{X_{\P}(y) < z_S}\mathbb{I}_{X_{\Q}(y) < z_S}\Phi_S(X_{\Q}(y) \vee X_{\P}(y) \vee y)}.
                \end{align*}

                If we observe the two expected revenues, we note that the last term is equal in both. Since we intend to compare the two revenues, this term may be ignored.
                Moreover, the dependence of all the stopping times is only on $y$, so we can omit it in the future. The remainder of this proof is similar in spirit to \cite{weitzman}. For simplicity of notation, we introduce the following shorthand:
                $$
                \pi_\P := \ind{\xp > z_\P}, \quad \lambda_\P := \ind{z_\P > \xp > z_S},  \quad \pi_\Q := \ind{\xq > z_\P},  \quad \lambda_\Q := \ind{z_\P > \xq > z_S}.
                   \\
                $$
                
                A few observations: $\pi_\P\lambda_\P=\pi_\Q\lambda_\Q=0$, {\em i.e.} the events are mutually exclusive, and
                $\pi_\P+\lambda_\P$ is the event that strategy A will stop after having explored $\P$ (by definition of strategy A). 
                With this notation, and ignoring the common term, we have:
               \begin{align*}
                   A = \E{-\cp + (\lambda_\P+ \pi_\P)\xp + (1-\pi_\P-\lambda_\P)(-\cq) + (1-\pi_\P-\lambda_\P)(\lambda_\Q + \pi_\Q)\xq},\\
                   B= \E{ - \cq + \xq \pi_\Q+\lambdaq(-\cp+\pp\xp+\lambdap (\xp \vee \xq)+\xq(1-\pp-\lambdap))}+ \mkern16mu\\ + \E{(1-\pq -\lambdaq)(-\cp + \pp \xp + \lambdap \xp)}.
               \end{align*}
                Computing the difference we get to:
                \begin{align*}
                   A - B = \E{\pq(\pp \xp - \cp)+(\lambdap + \pp)(\cq - \pq \xq)+(\lambdaq + \pq)\lambdap \xp - \lambdap \lambdaq (\xp \vee \xq)}.
                \end{align*}
                At this point we plug in the definition of reservation values using the independence of the two prefixes (they are independent because the strategies are designed in such a way that the stopping times on different prefixes are independent from the realization of the other one, given that they are both played): 
                \begin{align*}
                    \E{\cp}&=\E{(\xp-z_\P)_+}=\E{(\xp - z_\P)\pp} \\
                    \E{\cq}&\geq \E{(\xq - z_{\Q})\pq+(\xq-z_{\Q})\lambdaq} & (\ind{\xq \geq z_{\mathcal Q} } \geq \pi_{\Q}+\lambda_{\Q}).
                \end{align*}
                So we get:
                \begin{align*}
                    A&-B \geq \E{ \pq \pp (\zp-\zq) -\lambdap \pq \zq + \lambdaq(\lambdap + \pp)(\xq-\zq)}\\ 
                    &\mkern36mu+\E{(\pq + \lambdaq )\lambdap \xp-\lambdap \lambdaq (\xp \vee \xq)} \\ 
                    &\geq 
                    \E{\lambdap \pq (\xp - \zq) + \lambdaq \pp (\xq - \zq) + \lambdap \lambdaq (\xp + \xq - \zq - \xq \vee \xp)},
                \end{align*}
                where in the last inequality we used the fact that $\zp \geq \zq.$ 
                Recalling the definitions of the $\lambda$'s and $\pi$'s, 
                the only term that is not clearly positive is the rightmost term. 
                But $\xp+\xq-\xp\vee\xq = \xp\wedge\xq$, which is greater than $\zq$, unless $\lambdap\lambdaq=0$.
                Thus, we conclude $A-B\geq 0$, as desired.
                \\ \\
                \noindent \textbf{Case II:} If $z_1^1> y \geq z_1^i$, then consider the following modified instance: decrease the cost $c_1^i$ of box $b_1^i$ in such a way that the new $z_1^i$ value now lies (strictly) between $z_1^1$ and~$y$. Let $\Delta$ denote this change in cost.
                For the sake of mathematical analysis, we will allow negative cost.
                Denote as $\phi'(A)$, $\phi(B)$, and $\phi'(B)$ the expected performances of \textsc{Strategy A} on the modified instance, 
                \textsc{Strategy B} on the original instance, and \textsc{Strategy B} on the modified instance, respectively.
                Define \textsc{Strategy C} as making the same decisions as $A$ would in the modified instance --- as a function of the observed values --- but while playing in the original instance. Let $\phi(C)$ be its expected performance.
                
                Since B always opens $b_1^i$, we have that $\phi'(B)-\phi(B)=\Delta$.
                However, A can choose to open $b_1^i$, depending on the observed random variables, so $\phi'(A)-\phi(C)$ is an expectation over 0 and $\Delta$. 
                Thus, $\phi'(A)-\phi(C)\leq \phi'(B)-\phi(B)$. 
                However, we have shown in Case I that $\phi'(A)\geq\phi'(B)$, and so this implies that $\phi(C)\geq \phi(B)$. 
                Since C is suboptimal, this implies that B is suboptimal, as desired.
            \end{proof}

        \subsection{\PBPTree}\label{sec:tree-constraints-opt}

            In the last subsection we have solved the problem on multiple parallel lines. The main difficulty lied in the interplay between the different lines in the optimal strategy. 
            This was solved  by proving that the lines can be divided into {\em macroboxes} which behave  like single boxes.
            The work of the previous section is in fact enough to prove that the Generalized Pandora's Rule is optimal even for trees.
            
            We begin by extending the definition of reservation values beyond lines. 
            One might naturally try to extend stopping times to a more general exploration rule. 
            Instead, we observe that if there is an optimal threshold strategy on a subtree, then it is equivalent to view it as a line constraint, following the threshold ordering. 
            This highlights the power of the concept of macro-boxes: they are not only a feature of the analysis, but enable us to decompose a tree into a line.
            
            This requires the introduction of the function $\textsc{Merge}(\L_1,\L_2)$ which takes as input two lines along with the information on the reservation values of their boxes, and outputs the line obtained by their merging according to decreasing reservation values, maintaining the relative orders of boxes in the same line.
            With this in mind, the algorithm is presented formally in Algorithm \ref{alg:tree}.
            
            \def\term{\textsc{Terminate}}
            \def\null{\textsc{Null}}
            \begin{algorithm}
            \caption{Pandora's Rule for Tree}
                \KwData{Distributions of the random rewards, box costs and tree constraint $T$.}
                Initialize queue $leaves$;\\
                \For{$i\gets 1$ \KwTo $n$}{
                    \If{box $b_i$ is a leaf}{initialize line $\L_i=[b_i]$ and enqueue $b_i$ in $leaves$}}
                \While{$leaves$ is not the empty queue}{
                    dequeue box $b_\ell$ from $leaves$;\\
                    Initialize an empty line $\L$\\
                    \For{$b_j$ in $\textsc{children}(b_\ell)$}{
                        $\L \gets \textsc{Merge}(\L,\L_j)$\tcc*{Taking reservation values into account}
                    }
                    $z_\ell  \gets \textsc{computeThreshold}(b_\ell,\L)$ and $\L_\ell \gets [b_j] + \L$\\
                    \If{$\ell \neq 1$ and $z$ has been computed for all the children of $\textsc{parent}(b_\ell)$}{enqueue $\textsc{parent}(b_\ell)$ in $leaves$}
                }
                \textbf{return} the generalized reservation values $z_j$, and the linearized tree $\L_1$
            \label{alg:tree}
            \end{algorithm}
            Having defined the Generalized Reservation prices for the tree case, we show the following:
            \begin{thm}
                The Generalized Pandora's Rule is optimal for the \PBPTree and can be computed in polynomial time and space.
            \end{thm}
            \begin{proof}
                We will again do this by induction on the number of un-opened boxes. 
                If there is only one box, everything follows as before. 
                If we are in a state where there are multiple available subtrees to continue along,
                let $T_1,\,\dotsc,\,T_k$ be these subtrees, and $b_1,\,\dotsc,\,b_k$ be their respective roots.
                By induction, after opening any root box $b_i$, there is an ordering $\prec_i$ induced on the remaining boxes in all subtrees, for the optimal strategy to explore.
                Observe that restricted to any subtree $T_j$, the nodes of $T_j$ are ordered the same in all $\prec_i$ orderings, including $\prec_j$.
                Thus, we may define $\L_j$ as the line that represents this common ordering of the nodes of $T_j$; 
                it is clear by induction that the multi-tree problem on $T_1,\,\dotsc,\,T_k$ is no more profitable than the 
                multi-lines problem on $\L_1,\,\dotsc,\,\L_k$.
                Furthermore, the reservation value for $b_j$ at the head of $T_j$ is the same as the reservation value for entering $\L_j$, 
                again by induction. 
                Thus, this theorem is a corollary of the previous.
            \end{proof}
            
            \begin{remark}
                It is clear from the previous proof that the Generalized Pandora's Rule can be easily applied to Forest-Constrained instances. As we have already mentioned, it is sufficient to add a dummy root with no cost and no reward pointing to the roots of the actual trees in the forest to recover an equivalent Tree-Constrained Pandora's Box Problem. 
            \end{remark}
            

    \section{Adaptivity Gaps and Approximation Beyond Tree Constraints}\label{sec:appx-algs}
    In the previous section, we sought to design exactly optimal policies, and required exactly comparing the performance of alternative strategies. 
    As we will see in Section~\ref{sec:hardness}, we can not hope to do so for more general constraints, as the problem becomes NP-hard to 
    approximate. 
    For this reason, we seek instead to find approximately optimal solutions.
    We present in this section approximation algorithms for some cases of the \PBPDag.
    Following recent literature on stochastic probing \cite{singla2018price,gupta2017adaptivity,singla2018combinatorial,adaptivityGapProof}, we will go through an {\em adaptivity gap} route, arguing that for any adaptive strategy, there exists a non-adaptive strategy --- {\em i.e.} pre-computing a fixed set and opening it obliviously --- which approximates its performance.
    Therefore, the optimal non-adaptive strategy is a good approximation of the optimal adaptive strategy.
    
    However, as our setting is very broad, and captures many of the complexities of stochastic submodular optimization, it is not likely that an optimal non-adaptive set will be easy to find. 
    Instead, following an approach similar to \cite{anagnostopoulos2019stochastic}, we give a single adaptive strategy which performs, in expectation, better than {\em any} fixed set, and yields therefore a good approximation for the optimal adaptive strategy. 
    We begin with the adaptivity gaps. 
Recall the statement of Theorem~\ref{thm:adaptivity}:
{\renewcommand\footnote[1]{}\adaptivity*}
    
The proof of this theorem closely follows~\cite{adaptivityGapProof}, but it is short, and so we include it here for completeness. As~\citet{adaptivityGapProof} gives the proof in a more general setting,
this proof should in theory extend to other objective functions beyond the max-of-all-entries objective that we have been using. 

\begin{proof}
    The proof is a relatively simple, but clever, idea introduced in the adaptivity gap upper-bound of~\cite{adaptivityGapProof}.
    The idea is to show that if we choose a set at random, according to the distribution induced by $S(\pi^*)$ from the randomness on the rewards, then this set will perform well in expectation over both the random set, and the random rewards. 
    It follows that there must exist some set which performs at least as well as this in expectation.
    
    Formally, we wish to show that if we
    randomly sample the value of each box twice, choose optimal boxes adaptively for one of the samples, but measure revenue using the other samples, we lose only a factor 2 in the expectation of the $\max_{i\in S}X_i$ term. As for the $\sum_{i\in S}c_i$ term, we are opening the same set, so they cost the same.
    Note that this considers only feasible sets $S$, by definition.
    
    To this end, let $X_1,\,\dotsc,\,X_n$ be the random payoff values of the boxes, and let $Z_1,\,\dotsc,\,Z_n$ be respectively identically distributed copies of the $X_i$'s, sampled independently.  Fix an optimal adaptive strategy $\pi$, and let $\pi(S,y)\in [n]$ denote the choice of the next box to open after having opened $S$, and observing largest value $y$.
    Let $\mathbb S(\pi,X|S,y)$ be the (random) final set that $\pi$ opens when it chooses to terminate, if it starts with set $S$ and total $y$.
    We denote
    \[
        \mu_Z(S,y,y'):= \E{(-y'+\textstyle\max\{Z_i:i\in \mathbb S(\pi,X|S,y),\, i\notin S\})_+}
    \]
    the expected future gain when playing according to the $X_i$ values starting in state $(S,y)$, but measuring revenue with the $Z_i$'s from state $(S,y')$.
    Note that $\mu_X(\emptyset,0,0)$ is the expected revenue of playing according to the adaptive strategy, and 
    $\mu_Z(\emptyset,0,0)$ is the expected revenue of 
    randomly picking a set according to the $Z_i$'s. 
    
    We wish to show $\mu_X(S,y,y')\leq 2\mu_Z(S,y,y')$, by induction on the set $S$, as it ranges over all subsets, in decreasing order of cardinality.
    Note that if $(S,y)$ is such that the policy $\pi$ will choose to terminate, 
    then both values are $y-y'$. 
    Otherwise, fix $S$, $y$, and $y'$, and let $p:=\pi(S,y)$. We have
    \begin{align*}
        \nonumber \mu_X(S,y,y')&= \E{(X_{p}-y')_+ + \mu(S+p,\ y\vee X_{p},\ y'\vee X_{p})}\\
        &\leq \E{((X_{p}\vee Z_{p})-y')_+ + \mu(S+p,\ y\vee X_{p},\ y'\vee (X_p\vee Z_{p}))}\\
        &\leq \E{(X_{p}-y')_+ + (Z_{p}-y')_+ + \mu(S+p,\ y\vee X_{p},\ y'\vee (X_p\vee Z_{p}))}\\
        &= \E{2(Z_{p}-y')_+ + \mu(S+p,\ y\vee X_{p},\ y'\vee (X_p\vee Z_{p}))}
    \end{align*}
    Where the first inequality asserts that earning more up front an only help, and the last equality holds by linearity of expectation and the identical distributions of $X$ and $Z$. Furthermore,
    \begin{align*}
        \nonumber \mu_Z(S,y,y')&= \E{(Z_p-y')_+ + \mu_Z(S+p,\ y\vee X_p,\ y'\vee Z_p)}\\
        &\geq \E{(Z_p-y')_+ + \mu_Z(S+p,\ y\vee X_p,\ y'\vee (X_p\vee Z_p))}
    \end{align*}
    Since $\mu(S,y,y')$ is non-increasing in $y'$.
    By linearity of expectation, and by induction on $S$, we get $\mu_Z(\emptyset,0,0)\geq \tfrac12 \mu_X(\emptyset,0,0)$, as desired.
\end{proof}

With this result in hand, it remains to show that we can develop an {\em adaptive} strategy which performs at least as well as every non-adaptive strategy, in expectation. 
We will take advantage of the fact that we are working in a tree-constraint, and label the boxes with a pre-order of the nodes of the tree. We will denote the index of box $b$ as $i_b$. Recall that, by the properties of a pre-order, we have that for all $b$, if $b$ has $k$ descendants in the tree, then the descendants of $b$ are exactly those boxes indexed by $i_b+1,\,i_b+2,\,\dotsc,\,i_b+k$. This allows us to keep track of which boxes can legally be opened if we choose to not open $b$, since we may simply jump ahead in the pre-order. 

We wish to use this fact to design a simple dynamic program computing the best adaptive strategy among all which only consider boxes following the pre-order.
The pre-order allows us to use an index in the order to store the tree-constraint information on $S$, but it remains to efficiently encode information regarding the matroid constraint. 
To this end,
we define here a characterization of all constraints with ``oblivious feasibility oracles'': 
\begin{definition}
    A constraint on the feasible sets $S$ of boxes to open is said to have an {\em oblivious feasibility oracle} if it is characterized by a set function $D(S)$ with the following properties: 
    \begin{enumerate}
        \item $\{D(S):S\subseteq [n]\}$ is supported on  polynomial in $n$ values,
        \item  For any $S$ and $u\notin S$, $D(S+u)$ is efficiently computable knowing only $D(S)$ and $u$, and
        \item For any $S$, it can be efficiently determined whether $S$ is feasible knowing only $D(S)$.
    \end{enumerate}
\end{definition}

To illustrate this notion, we take as an example a {\em generalized knapsack constraint}, where every box $b$ is assigned a vector $\bm w_b\in \mathbb Z_+^d$, and we have a capacity vector $\bm m\in \mathbb Z_+^d$. Here $d$ is a constant.
A set $S$ is feasible if $\sum_{b\in S} \bm w_b\preceq \bm m$, taken componentwise. 
The function $D(S)$ is simply $\sum_{b\in S} \bm w_b$, and 
we require that the entries of $\bm m$ be polynomial in $n$.

Note that this generalized knapsack constraint includes, as a special case, knapsack constraints, cardinality constraints, and even partition matroids with $O(1)$ partitions.

We will define the function $\Psi(i,y,D)$ recursively below, which denotes the expected revenue if we start at position $i$ in the sequence, having already collected $y$, with feasibility oracle value $D$. As a base case, $\Psi=0$ when $D(S)=D$ implies $S$ is not feasible, and $\Psi=y$ when $i=n+1$. Otherwise, 
let $\textsc{next}(i)$ denote the first position after $i$ in the pre-order on the tree such that $\textsc{next}(i)$ is not a descendent of $i$. Then
\[
    \Psi(i,y,D(S)) := \max \begin{cases}
        y\\
        \Psi(\textsc{next}(i),y,D(S))\\
        -c_i + \E{(X_i-y)_+} + \E{\Psi(i+1,\ y\vee X_i,\  D(S+i))}
    \end{cases}
\]
Since the $X_i$'s take only polynomially many values, then this function can be computed in polynomial time, by definition of $D$. 
We can also simultaneously compute the associated adaptive policy $\pi$ as in Algorithm~\ref{alg:adapt-approx} below. 
Let $Y$ be the set of all possible values attained by all the $X_i$'s, and $\mathcal D$ be all possible values attained by $D(S)$.
The $\theta$ function returned by the algorithm determines the policy: if we are in state $(S,y)$, and the max index of an element in $S$ is $i$, then $\pi(S,y):= \theta(i,y)$.

\def\term{\textsc{Terminate}}
\def\null{\textsc{Null}}
\begin{algorithm}
\caption{Approximately Optimal Adaptive Strategy}
    \KwData{Pre-ordering $b_1,\,\dotsc,\,b_n$, Oblivious feasibility oracle $D$, box costs, and random payoffs.}
    \For{$i\gets n+1$ \KwTo $1$}{
        \For{$y\in Y$, $D\in \mathcal D$}{
            \lIf{$D$ is infeasible}{$\Psi(i,y,D)\gets 0$}
            \lElseIf{$i=n+1$}{$\theta(n,y)\gets \term$ \textbf{and} $\Psi(n+1,y,D)\gets y$}
            \Else{
                $\textsc{open}\gets -c_i + \E{(X_i-y)_+} + \E{\Psi(i+1,\ y\vee X_i,\  D(S+i))}$\\
                $\textsc{skip}\gets \Psi(\textsc{next}(i),y,D)$\\
                $\Psi(i,y,D)\gets \max\{y,\textsc{open},\textsc{skip}\}$\\
                \lIf{$\max\{y,\textsc{open},\textsc{skip}\} = y$}{$\theta(i-1,y)\gets \term$}
                \lIf{$\max\{y,\textsc{open},\textsc{skip}\} = \textsc{open}$}{$\theta(i-1,y)\gets i$}
                \lIf{$\max\{y,\textsc{open},\textsc{skip}\} = \textsc{skip}$}{$\theta(i-1,y)\gets \theta(\textsc{next}(i),y)$}
            }
        }
    }
\label{alg:adapt-approx}
\end{algorithm}
\begin{claim} \label{obs:best}
 The strategy returned by Algorithm \ref{alg:adapt-approx} is at least as good as any non-adaptive strategy.\end{claim}
 \begin{proof}
    This can be seen by induction on the $i$ variable of the dynamic program. Let $S_{-j}:=S\cap\{n-j,\,\dotsc,\, n\}$.
    We wish to show that $\Psi(n-j,y,D(S\setminus S_{-j}))\geq \E{\max_{i\in S_{-j}}X_i - c(S_{-j})}$  for all $y$, by induction on $j$. Note that for $j=0$,  both values are equal to the revenue of set $S$.  
    For $j>0$, regardless of whether $n-j+1\in S$, $\Psi$ takes the max over including it and not including it, and by induction, the following $\Phi$ term performs better than $S_{-(j-1)}$ in expectation.
 \end{proof}
    Combining \Cref{obs:best} with \Cref{thm:adaptivity} gives us the following result:

\begin{thm}
    For the \PBPTree augmented with oblivious-feasibility-oracle matroid constraints, we can efficiently compute a policy $\hat \pi$ such that for any $\pi$,
    \begin{equation*}\textstyle
    \E{\max_{i \in S(\hat \pi)} X_i - \sum_{i \in S(\hat \pi)}c_i} \geq \tfrac{1}{2} \E{\max_{i\in S(\pi)} X_i} - \E{\sum_{i\in S(\pi)} c_i}
    \end{equation*}
\end{thm}

    \section{Impossibility and Hardness Results}\label{sec:hardness}

In this Section, we show the impossibility results outlined in Section \ref{sec:model}. 
We first show that, when the precedence graph is not a tree, then there may not exist an optimal strategy which has a threshold structure. 
We then show the approximation hardness of solving the \PBP with both general order constraints and $\F= \F_T \cap \mathcal{I}_\M$ where $T$ is a tree and $\M$ is a matroid.

The \PBPDag where constraints are given by a DAG $G$ ,
requires that a box only be opened once one at least one of its in-neighbours in $G$ is open.

\subsection{Suboptimality of Threshold Strategies}

    \begin{thm}
        The \PBPDag need not admit an optimal threshold strategy, when the constraint graph is not a tree. Moreover the same holds for constraints $\F_T \cap \mathcal{I_{\M}}$, where $T$ is a tree and $\M$ a matroid.
    \end{thm}
    \begin{proof}
        Consider graph (a) in Figure \ref{fig1} with the following parameters:
        \begin{equation*}
        \arraycolsep=1.4pt\def\arraystretch{1.25}
        \begin{array}{llllllll}
            X_A &= \left\{\begin{array}{ll}
                2.5&\text{ w.p. }\tfrac 12\\
                0&\text{ w.p. }\tfrac 12
            \end{array}\right.,\quad
            &X_B &= 2,\quad
            &X_C &= \left\{\begin{array}{ll}
                3&\text{ w.p. }\tfrac 12\\
                0&\text{ w.p. }\tfrac 12
            \end{array}\right.,\quad
            &X_D &= \left\{\begin{array}{ll}
                6&\text{ w.p. }\tfrac 12\\
                0&\text{ w.p. }\tfrac 12
            \end{array}\right.,
            \\[1.2em]
            c_A&=0,&
            c_B&=1,&
            c_C&=1-\tfrac\varepsilon 2,&
            c_D&=0.
        \end{array}\end{equation*}

        For $\varepsilon \in [\tfrac54,2]$ it can be shown that it is optimal to start the exploration of the graph from $A$, then, depending on the realization of $X_A$ it is optimal to open $B$ (and then $D$) or to open $C$ (then $D$ and then possibly $B$). 
        If we now consider an instance of the \PBPDag on $\F_T\cap \mathcal I_\M$ where $\mathcal I_\M$ is all subsets of cardinality 4, and $T$ is given by (b) in Figure \ref{fig1}, with boxes $A,B$ and $C$ and two copies $E$ and $F$ of $D$, then we inherit the results from (a).
    \end{proof}

    \begin{figure}[h!] 
    \def\r{0.8}
    \centering
    {\footnotesize
    \begin{tikzpicture}[node distance = {1.0cm and 1.0cm}, v/.style = {draw, circle}, pil/.style = {->}]

        \node (text) at (-1.65*\r,1.2*\r) {(a)};
        \node (a) [v] at (0,\r) {A};
        \node (c) [v] at (\r,0) {C};
        \node (b) [v] at (-\r,0) {B};
        \node (d) [v] at (0,-\r) {D};
        
        \draw[->] (a) to node[above]{} (c);
        \draw[->] (a) to node[above]{} (b);
        \draw[->] (b) to node[below]{} (d);
        \draw[->] (c) to node[below]{} (d);

    \end{tikzpicture}\hspace{3em}
    \begin{tikzpicture}[node distance = {1.0cm and 1.0cm}, v/.style = {draw, circle}, pil/.style = {->}]

        \node (text) at (-1.65*\r,1.2*\r) {(b)};
        \node (a) [v] at (0,\r) {A};
        \node (c) [v] at (\r,0) {C};
        \node (b) [v] at (-\r,0) {B};
        \node (e) [v] at (-\r,-1.1*\r) {E};
        \node (f) [v] at (\r,-1.1*\r) {F};
        
        \draw[->] (a) to node[above]{} (c);
        \draw[->] (a) to node[above]{} (b);
        \draw[->] (b) to node[below]{} (e);
        \draw[->] (c) to node[below]{} (f);

    \end{tikzpicture}}
    \caption{The order of optimal adaptive exploration is not fixed}
    \label{fig1}
    \end{figure}
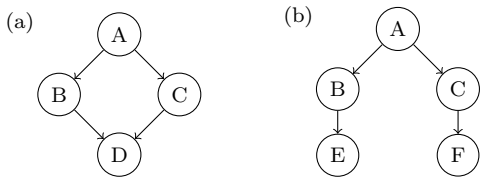
    \subsection{Hardness of Approximation}
    As previously mentioned, we wish to show that it is NP-hard to approximate an optimal strategy for general order constraints and $\F= \F_T \cap \mathcal{ I}_{\M}$ where $T$ is a tree and $\mathcal{ I}_{\M}$ are the independent sets of a matroid $\M$.
    Formally, we prove it is NP-hard to design a policy with approximately optimal rewards, for some constant. 
    Approximation is taken in the sense of the previous section. 
    \begin{thm}
    \label{thm:daghardness}
        It is NP-hard to approximate within 0.9997 the optimal strategy to the \PBP with DAG constraints.
        It is sufficient for the DAG to have depth 2 and fan-in 3.
    \end{thm}
    
    We will be reducing from the problem of finding a minimum vertex cover on cubic graphs, which is known to be hard to approximate.
    We argue here that, since the constraint graph $G$ has depth $2$ and fan-in $3$, this also implies hardness for tree-and-matroid constraints, $\F_T\cap \mathcal I$. 
    
    \begin{corr}
    It is NP-hard to find the optimal strategy to the \PBPDag with constraint $\F=\F_G \cap \mathcal{ I}_{\M}$ where $\M$ can be any matroid on $B$ even if $G$ is restricted to be a tree.
    It suffices for $\M$ to be a partition matroid.
    \end{corr}
    \begin{proof}
    The constraint graph $G$ from Theorem~\ref{thm:daghardness} consists of sources and sinks, such that every sink had exactly 3 sources as its parents. 
    Replace each sink with 3 identical copies (including costs and rewards), assign one to each parent, and restrict that at most one copy of each is opened. 
    This is exactly a partition matroid constraint, and the resulting graph is a forest of depth-2 trees, and is equivalent in terms of exploration costs and rewards to the constraint graph $G$.
    \end{proof}

    Finally, the proof of the Theorem is given below.
    
    \begin{proof}[Proof of Theorem~\ref{thm:daghardness}]
    It is known that it is NP-hard to approximate the minimum vertex cover of a cubic graph within a factor of $\approx 1.0012=:1+\epsilon_0$~\cite{clementi1999improved,papadimitriou1991optimization}.
    Let $G=(V,E)$ be a hard-to-approximate instance, 
    and let $n:=|V|$, and $m:=|E|$. 
    Let $\alpha$ be such that the optimal vertex cover has size $\alpha m$.
    Observe, since $G$ is cubic, that $m:=\tfrac32n$, and $\alpha\geq \tfrac 13$. 
    Furthermore, any greedy independent set must have at least $\tfrac 1{3+1}n$ nodes, which implies that its complement is a vertex cover of size at most $\tfrac 34 n = \tfrac 12 m$. Thus, $\alpha\in[\tfrac13,\tfrac 12]$.
    
    We construct, now, the constraint graph $D$. 
    The nodes of $D$ will be labelled by $V\cup E$, where the $V$ nodes will be the sources of the DAG, each having cost $1$ and reward $0$, and the $E$ nodes will be the sinks of the DAG, each having cost $0$ and reward $\beta m$ with probability $\tfrac c m$ and 0 otherwise, for constants $\beta,c>0$ which we will choose later. 
    There is an edge connecting any vertex-box $v$ to each edge-box $e$ such that $e$ in incident to $v$. Since $G$ is cubic, this implies that $D$ has depth 2 and fan-in 3, as required in the theorem statement.
    
    Any optimal strategy must take the following form: (1) Fix an ordering on the boxes labelled by $V$, (2) Pay to open the next vertex-box in the order, then reveal the $\leq 3$ unopened edge-boxes which it reveals. (3) Repeat until the reward has been collected. 
    Suppose that the $i$-th vertex-box we pay for allows us to open $0\leq n_i \leq 3$ new edge-boxes, and $N_i:=\sum_{j=1}^{i-1} n_j$. 
    Then the expected max reward will be $\beta m \cdot (1- (1-\tfrac cm)^m)$, and the expected cost will be
    \[
        \E{\text{\# $V$ boxes opened}}\ =\ \sum_{i=1}^n \Pr{\text{opening $\geq i$ boxes}}\  =\  \sum_{i=1}^n (1-\tfrac cm)^{N_i}
    \]
    
    Observe, without loss of generality, $n_{i-1}\geq n_i$ for all $i$, as swapping the $(i-1)$-st and $i$-th boxes will only increase $N_i$ and leave $N_{i+1}$ and onwards unchanged, reducing the expected cost.
    Thus, in any fixed order, after this swapping, there must exist numbers $k_3$, $k_2$, and $k_1$, such that 
    \[
        n_1=n_2=\dotsm = n_{k_3} = 3,\quad 
        n_{k_3+1}=n_{k_3+2}=\dotsm = n_{k_3+k_2} = 2,\quad
        n_{k_3+k_2+1}=0\dotsm = n_{k_3+k_2+k_1} = 1
    \]
    Note that $3k_3+2k_2+k_1=m$, and that the vertex cover has size $k_3+k_2+k_1$. 
    Setting $r=(1-\tfrac cm)$, we have that the expected cost becomes
    \[
        \sum_{i=1}^{k_3}r^{3i} + \sum_{i=1}^{k_2}r^{3k_3+2i} + 
        \sum_{i=1}^{k_1}r^{3k_3+2k_2+i} = \tfrac{r^3}{1-r^3}
        \left(1-r^{3k_3}\right)
        +
        \tfrac{r^2}{1-r^2}
        \left(r^{3k_3}-r^{3k_3+2k_2}\right)
        +
        \tfrac{r}{1-r}
        \left(r^{3k_3+2k_2}-r^{m}\right)
    \]
    
    In the remainder of the proof, we will bound the values of $k_3$, $k_2$, and $k_1$, for optimal and sub-optimal vertex covers, and show that the difference in expected cost is at least a constant factor of the expected reward. 
    Since it is NP-hard to approximate the vertex cover, this will imply that is it NP-hard to approximate the optimal strategy for the \PBP on $D$. 
    
    Let $S^*$ be an optimal vertex cover of size $\alpha m$, and let $S'$ be any vertex cover of size $\geq (1+\epsilon_0)\alpha m$.
    Let $k_3^*$, $k_2^*$, and $k_1^*$, be as above for the set $S^*$, and $k_3'$, $k_2'$, and $k_1'$ be similarly for $S'$. 
    We wish to lower-bound the cost of opening $S^*$, and upper-bound the cost of opening $S'$, by bounding the possible values of the $k^*$'s and $k'$'s obtained by sub- and super-optimal orderings, respectively. 
    Note that we can trade off $k_3+k_1$ for $2k_2$ to increase the expected cost, and vice versa. 
    Since $\alpha \leq \tfrac12$, then for $S^*$, it will suffice to assume $k^*_1=0$, and increase $k_3^*$ as $\alpha$ approaches $\tfrac 13$. 
    For $S'$, it will suffice to assume $k_2'=0$.
    With the constraints on the $k^*$'s and the $k'$'s, this gives
    \[
        \begin{cases}
            k_3^* = (1-2\alpha) m\\
            k_2^* = (3\alpha -1) m
        \end{cases} \qquad 
        \begin{cases}
            k_3' = \tfrac12(1-(1+\epsilon_0)\alpha) m\\
            k_1' = \tfrac12(3(1+\epsilon_0)\alpha -1) m
        \end{cases}
    \]

    Furthermore, in the expected cost expression above, we get 
    \begin{align}\label{eq:cost1}
        \E{\text{cost}(S^*)}&\leq
        \tfrac {r^3}{1-r^3}\left(1 - r^{3k^*_3}\right) + 
            \tfrac {r^2}{1-r^2}\left(r^{3k^*_3} - r^{m}\right)= \tfrac{r^3}{1-r^3} - \tfrac{r^2}{1-r^2} (r^m) + 
            r^{3k^*_3}\left[
            \tfrac{r^2}{1-r^2} - \tfrac{r^3}{1-r^3}
            \right]\\ \label{eq:cost2}
        \E{\text{cost}(S')}&\geq \tfrac {r^3}{1-r^3}\left(1 - r^{3k'_3}\right) + 
            \tfrac r{1-r}\left(r^{3k_3'} - r^{m}\right)
            = \tfrac{r^3}{1-r^3} - \tfrac{r}{1-r} (r^m) + 
            r^{3k_3'}\left[
            \tfrac{r}{1-r} - \tfrac{r^3}{1-r^3}
            \right]
    \end{align}
        
        Combining (\ref{eq:cost1}) and (\ref{eq:cost2}), the difference $\Delta:=\E{\text{cost}(S') - \text{cost}(S^*)}$ is at least
        \begin{align*}
            \Delta&\geq -r^m\left[ \tfrac{r}{1-r}-\tfrac {r^2}{1-r^2}\right]
            - \tfrac{r^3}{1-r^3}\left(r^{3k_3'}-r^{3k_3^*}\right) + \tfrac{r}{1-r}r^{3k_3'} - \tfrac{r^2}{1-r^2}r^{3k_3^*}\\
          &= \left(r^{3k_3'}-r^{3k_3^*}\right)\left[\tfrac{r^2}{1-r^2}-\tfrac{r^3}{1-r^3}\right] + \left[\tfrac{r}{1-r^2}\right]\left(r^{3k_3'} - r^{m}\right)
        \end{align*}
        Recalling the values of $k_3^*$ and $k_3'$ above, expanding $r=1-\tfrac cm$, and first taking MacLaurin series around ``$\tfrac cm$''$=0$ for the terms in square brackets, then Taylor series for the terms in round brackets, we have
        \begin{align}
            \nonumber \Delta&\geq (\tfrac m{6c} + O(1))\left((1-\tfrac cm)^{3k_3'}-(1-\tfrac cm)^{3k_3*}\right) + (\tfrac m{2c} + O(1))\left((1-\tfrac cm)^{3k_3'}-(1-\tfrac cm)^m\right)\\
            \nonumber &= \tfrac m{6c}\left(e^{-3c(1-\alpha-\epsilon_0\alpha)/2}-e^{-3c(1-2\alpha)}\right) + \tfrac m{2c}\left(e^{-3c(1-\alpha-\epsilon_0\alpha)/2}-e^{-c}\right)+O(1)\\
            &=\tfrac m{6c}\left(4e^{-3c(1-\alpha-\epsilon_0\alpha)/2}-3e^{-c}-e^{-3c(1-2\alpha)}\right)+O(1)
        \end{align}
        Setting $c=(2\epsilon_0)/(3\alpha)$,
        and for convenience, denoting $A:=1/\alpha$, we have
        \begin{equation}
            \Delta\ \geq \ \tfrac{m}{4A\epsilon_0} \left(4e^{-\epsilon_0(A-1-\epsilon_0)}-3e^{-c}-e^{-3c(1-2\alpha)}\right)+O(1)
        \end{equation}
    Recalling that $A\in [2,3]$. 
    For $\epsilon_0=0.0012$ as in~\cite{clementi1999improved,papadimitriou1991optimization}, it can be shown that the function is non-increasing in $A$ on its domain, and plugging $A=3$, 
    we numerically have\footnote{
    More generally, taking the second derivative in $\epsilon_0$ suffices to show that the right hand side is strictly convex in $\epsilon_0$, since $A\in [2,3]$, and its derivative is 0 when $\epsilon_0=0$. This ensures the constant is a positive function of $\epsilon_0$. } $\Delta\geq 0.000399 \cdot m$.
    
    It remains to determine the ratio of the difference in expected cost to the expected reward. Recall that we have set the reward to be $\beta m$ with probability $\tfrac cm$, and 0 otherwise. 
    Since it costs 1 to open a box, and we wish to ensure that even when there is a single edge-box remaining, it is in our interest to open the box, we must set $\beta = \tfrac 1c$. Recall, then, that the expected reward will be
    \[
        \beta m\cdot (1-(1-\tfrac cm)^m) = m\cdot \tfrac 1c(1-e^{-c}+O(\tfrac 1m))= m + O(1) 
    \]
    Thus, an approximation for the \PBP which additively approximates the cost within a $1.00039$ factor of the revenue implies an approximation algorithm for vertex cover on cubic graphs within a factor of $<1.0012$, which is not possible unless $P=NP$. 
    This concludes the proof with a multiplicative constant of $1-0.00039<0.9997$ in the sense of Definition~\ref{def:appx}.
    \end{proof}

\section{Conclusion and Further Directions}
We have shown that solving the \PBPDag admits an efficiently computable optimal solution, for tree-like order constraints. We further showed that unless P=NP, there is no PTAS when the constraints are slightly generalized, and complement this result by showing an approximation algorithm for oblivious matroid constraints on top of tree-like precedence constraints. 
This latter result was shown by upper-bounding the adaptivity gap, and giving methods for beating optimal non-adaptive strategies. 
It is clear then that extending this problem to more general constraints, or even more general objective functions, can be done by giving more general approximation algorithms for the non-adaptive problem. 

\section*{Acknowledgments}
    Federico Fusco, Philip Lazos and Stefano Leonardi are partially supported by ERC Advanced Grant 788893 AMDROMA ``Algorithmic and Mechanism Design Research in Online Markets'' and MIUR PRIN project ALGADIMAR ``Algorithms, Games, and Digital Markets''.
    At the time of writing, Shant Boodaghians was visiting Stefano Leonardi as well.
    Shant Boodaghians is also partially supported by NSF grant 1750436.

\bibliographystyle{plainnat}
\bibliography{bibliography}

\clearpage

\appendix

\section{Learning Pandora}
\label{app:learning}
    
        In this section we apply the techniques in \cite{guo2019generalizing} to prove that a polynomial number of samples from the random variables $\{X_i\}_{i=1}^n$ is enough to solve with a good approximation the \PBPDag generalizing to any prefix-closed constraint.
        Furthermore we prove that, for the \PBPTree, a linear number of samples is enough and it is tight.
        
        The learning procedure to achieve this goal is indeed quite simple: 
        For a fixed $\epsilon>0$, we take an $O(\epsilon)$-grid of the interval $[0,1]$, and 
        for sufficiently many samples, learn the empirical distribution on these grid points.
        We then compute the optimal search policy using this empirical distribution.
        That this is an $\epsilon$-approximation is a straightforward application of standard techniques, and a proof is given in Appendix C2 of~\cite{guo2019generalizing}.
        Formally,
        \begin{itemize}
            \item For each box $b$, let $X_b^\epsilon$ be the random variable obtained by rounding down the reward $X_b$ to the nearest multiple of $c\cdot \epsilon$;
            \item Given $N$ {\em i.i.d.} samples of $X_b^\epsilon$, let $\hat X_b$ be the random variable distributed according to the empirical distribution;
            \item Output the strategy $\hat \pi$ which is optimal with respect to the $\hat X_b$'s.
        \end{itemize}
        We will show that for constants $\epsilon,\delta>0$, and a sufficiently large $N$ depending on $\epsilon$ and $\delta$,
        $\hat \pi$ will be an additive $\epsilon$-approximation of the optimal policy with probability $1-\delta$.
        
        \paragraph{General Constraints.}
        We recall that we have assumed rewards and costs are bounded in $[0,1]$. It is an immediate corollary of Theorem 1 and 7 in~\cite{guo2019generalizing} that, for the \PBP with {\em any} constraints, it suffices to have 
        \[
            N \geq C_1\cdot \tfrac {n^3}{\epsilon^3}\log(\tfrac n{\epsilon\delta})
        \]
        for some universal constant $C_1>0$.
        
        \paragraph{Tree Constraints.}
        We extend the techniques of~\cite{guo2019generalizing} to show that linearly many samples are sufficient for tree constraints. 
        The $n^3$ term in the previous bound comes from the fact that, when rewards and costs are bounded in $[0,1]$, the total performance of a strategy must lie in $[-n,1]$, requiring the $\epsilon$ value to be normalized by $n$. They use more specialized concentration bounds to get around this issue, which we extend to our setting.
        
        The goal is to show that, for an optimal algorithm, the performance over time forms a {\em submartingale}. Equivalently, one should only open a box if, in expectation, the revenue is increasing. 
        This is not true at face value, as it is often necessary to open bad boxes to allow us to move onto better boxes.
        We will use the notion of {\em macro-boxes} which were used in the proof of \Cref{main-thm}, where we showed the generalized Pandora's Rule was optimal for the \PBPTree. 
        \begin{definition}[Macro-Boxes]\label{def:macro}
            Let $\L_T=(b_1,\,b_2,\,\dotsc,\,b_n)$ denote the optimal order of exploration given by the Generalized Pandora's Rule on $T$, and assume $z_i$ is the generalized reservation value for box $b_i$ in this order. 
            Construct a sequence of indices as follows:
            $j_1=1$, and for all $i\geq1$, $j_{i+1}$ is the first index $j>j_i$ such that $z_{j}\leq z_{j_i}$. 
            Then we say that the $i$-th {\em macro box} is given by the collection of boxes $\{b_{j_i},\,b_{j_i+1},\,\dotsc,\, b_{j_{i+1} -1}\}$.
        \end{definition}
        \begin{claim}
             Let $j_1,\,\dotsc$ be as in \Cref{def:macro}, 
             Let $S_i(\pi^*)$ denote the (random) set obtained by following the optimal policy $\pi^*$ only until index $j_i$, and 
             define the random variable 
             \[\textstyle M_i:= \max_{\ell \in S_i(\pi^*)}X_\ell-\sum_{\ell\in S_i(\pi^*)}c_\ell.\] 
             Then the $M_i$'s form a submartingale.
        \end{claim}
        \begin{proof}
            We must show that for all $i\geq 1$, 
            \[
                \E{M_{i+1}|M_i,\,\dotsc,\,M_1}\geq M_i.
            \]
            By \Cref{prop01}, and by the definition of the $j_i$'s, we have that $z_{j_i},\,\dotsc,\,z_{j_{i+1}-1}$ remain unchanged if we truncate the sequence $\L_T$ to end at $b_{j_{i+1}-1}$. But in this case, $M_{i+1}$ is simply the performance of $\pi^*$ on the whole set. 
            By definition of reservation values, we then have that $\E{M_{i+1}|M_i}\geq M_i$.
        \end{proof}
        
        Another necessary condition for the result in~\cite{guo2019generalizing} is the {\em strong monotonicity} of the problem. 
        \begin{definition}[First-Order Stochastic Dominance]
            Random vector $\bm X'$ {\em stochastically dominates} $\bm X$ if, for every component $i$, and every $x\in \R$, we have $\Pr{X'_i\geq x}\geq \Pr{X_i\geq x}$.
        \end{definition}
        \begin{definition}[Strong Monotonicity]
            A problem is {\em strong monotone} if for any random variable $\bm X$, and any random variable $\bm X'$ which dominates $\bm X$, 
            Letting $\pi^*$ be the optimal policy for the distribution on $\bm X$,
            we have that the performance of $\pi^*$ over $\bm X'$ is at most the performance over $\bm X$.
        \end{definition}
        The following is a direct corollary of Appendix C3 in~\cite{guo2019generalizing}, when viewed over the macro-boxes, as we have a fixed order of exploration, and reservation prices.
        
        \begin{claim}The \PBPTree is strongly monotone.
        \end{claim}
        
        These two previous claims imply that Lemma 25 in~\cite{guo2019generalizing} apply.
        Thus, for the \PBPTree, 
        it suffices to take 
        \[ 
            N\geq C_2\cdot \tfrac n{\epsilon^2}\log^2(\tfrac 1\epsilon)\log(\tfrac n \epsilon)\log(\tfrac n{\epsilon\delta})
        \]
        for some universal constant $C_2$.

        \paragraph{Lower Bounds.}
        This latter result is tight up to $\operatorname{poly}\log(\tfrac n{\epsilon\delta})$ terms: \cite{guo2019generalizing} show that it takes at least $\Omega(\tfrac n{\epsilon^2})$ samples to get the desired degree of accuracy.

\bigskip
\section{The Adaptivity Gap of the \PBP is unbounded}
\label{app:counterexamples}
    In this section we present a counterexample showing that the adaptivity gap for the \PBP is unbounded. 
    This is inspired by an example found in~\cite{singla2018combinatorial}.
    
    Let $p>0$, and consider $n$ identical boxes with cost $c=1-p/2$, and reward $\tfrac 1{p^2}$ with probability $p^2$, 0 otherwise. 
    Since $c<1$, the adaptive optimal strategy is to open boxes until you get the reward, which guarantees reward $\tfrac 1{p^2}$ and costs $\tfrac c{p^2}$ in expectation,
    for a total expected revenue of $\tfrac1{p^2}(1-c) = \tfrac 1{2p}$
    
    We now consider the non-adaptive strategy which opens $k$ boxes.
    It earns $\tfrac 1{p^2}$ with probability $1- (1-p^2)^k$ and 0 otherwise, and pays $ck$.
    Note that $(1-p^2)^k$ is convex in $k$, and so $\tfrac 1{p^2}(1-(1-p^2)^k)$ is concave, and thus has non-increasing derivatives. 
    At $k=\tfrac 1p$, we have
    \begin{align}
        \nonumber \frac{\mathrm d}{\mathrm dk} \left[-ck+\tfrac 1{p^2}(1-(1-p^2)^k)\right] &= -c-(1-p^2)^k\cdot \frac{\ln(1-p^2)}{p^2} \\
        \label{eq:bad-ex-line-2}&\leq -c+\frac{-(-p^2)+p^4}{p^2}\cdot (1-p^2)^k &\text{if }p^2\leq \tfrac12\\
        \nonumber &\leq -c+(1+p/8)(1-p^2)^k&\text{if }p\leq \tfrac12\\
        \nonumber &\leq -c+(1+p/8)e^{-p}\\
        \label{eq:bad-ex-line-4} &\leq -(1-p/2)+(1+p/8)(1-0.632 p) \\ 
        \nonumber &\leq (\tfrac 12 - 0.632 + 0.125)p = -0.007p<0
    \end{align}
    Where \eqref{eq:bad-ex-line-2} holds because $\ln(1+x)\geq x-x^2$ for $x\in [-\tfrac 12,0]$: at $x=0$, $\ln(1+x)=x-x^2 = 0$, and $\tfrac {\mathrm d}{\mathrm dx}\ln(1+x) = \tfrac 1{1+x}\leq\tfrac {\mathrm d}{\mathrm dx} x-x^2 = 1-2x$ over the domain, by convexity of the former.
    \eqref{eq:bad-ex-line-4} holds since $e^{-p}\leq 1-(1-1/e)p$ for $p\in [0,1]$, again by convexity. 
    Hence, the derivative is negative, and we conclude the optimum is attained on $1\leq k <1/p$. 
    However, $(1-p^2)^k\geq 1-kp^2$, so
    \begin{align}
        \nonumber -ck+\tfrac 1{p^2}(1-(1-p^2)^k)&\leq -ck+\tfrac 1{p^2}kp^2 = (1-c)k
    \end{align}
    Since $1-c=p/2$, and $k\leq 1/p$, this upper-bounds the total revenue by $\tfrac 12$.
    
    Recalling that the adaptive strategy earned on average $1/2p$, 
    then the ratio of the two revenues is $\tfrac 1p$. 
    Taking $p\to 0$, this suggests that the adaptivity gap can be arbitrarily large. \qed
    
    \bigskip
    We remark that this counterexample works also for the single line constrained case as the boxes are all equal and the order is irrelevant.

    \bigskip
\section{Proof of Claim~\ref{cl:reservazionLine}}
\label{app:proofs}

        Recall, we have denoted the expected future reward following $\tau$, given that the player has just opened box $b_{i-1}$ and the largest reward that has been sampled in the past has been $x$, as:
        \[
        \phi^{\tau}(x,i)\ :=\ \E{\max\left\{x,\ \textstyle\max_{j=i}^{\tau(x,i)}X_j\right\}-\textstyle\sum_{j=i}^{\tau(x,i)} c_j},
        \]
        and denoted  $\Phi(x,i)=\max_{\tau}\phi^{\tau}(x,i)=\phi^{\tau^*(x,i)}(x,i)$, for all $0\leq i\leq n$. Observe that $\Phi$ is a non decreasing function in its first argument.
            
            \bigskip
            For the first part of the Claim we are going to prove that $\forall i=1, \dots, n$ the functions 
            \begin{equation}
                \label{eq:welldefined}
                H_i(x)=\E{\left( \max_{j=i}^{\tau^*(x,i)}X_j-x \right)_+-\sum_{j=i}^{\tau^*(x,i)}c_j}=\Phi(x,i)-x,
            \end{equation} 
            admit at least one zero each in $[0,\infty)$ and that the each solution set admits a minimum. In order to do so it is sufficient to show that the $H_i$ are continuous and monotone non increasing: it is then straightforward to conclude, since $H_i(0)=\Phi(0,i)\geq 0$ and $\lim_{x \to \infty}H_i(x)=0.$ \\
            For any two numbers $b\geq a \geq 0$ and $i$ we have that 
            \begin{align}
            \label{eq:lipshitz}
                H_i(b)-H_i(a)=\Phi(b,i)-\Phi(a,i)+a-b 
                 \leq  \E{\max_{j={i}}^{\tau^*(b,i)}X_j \vee b -\max_{j={i}}^{\tau^*(b,i)}X_j \vee a} - b+a \leq 0,
            \end{align}
            where in the first inequality we used the fact that the stopping time $\tau^*(b,i)$ is sub-optimal for $\phi^{\tau}(a,i)$. In particular (\ref{eq:lipshitz}) means that $\Phi(\cdot,i)$ is $1$-Lipschitz and hence continuous. So we can claim the continuity of all the $H_i(\cdot)$ as compositions of continuous functions.
            
            We also argue that from (\ref{eq:welldefined}) it is clear that the $H_i$ (and hence the solutions $z_i$) do not depend on the particular stopping time used in the definition (which may not be unique), but only on the optimal values $\Phi(x,i)$, which are unique.
            
            \bigskip
            Let's now focus on the second part, we want to prove that for every $i$ there exists a $\tau^*(z_i,i)$ which is different from $i-1$, as long as $z_i>0$. We recall that we are only considering stopping times that depends deterministically on the realizations, hence the event $\tau(x,i)$ opens or not box $b_i$ has either probability $1$ or $0.$
            Let's fix any $i$ for which $z_i>0$ and let $\overline \tau_i$ be the best strategy between all those that open box $b_i$, we want to prove by contradiction that 
            $$
                \delta=\Phi(z_i,i)-\phi^{\overline{\tau}_i}(z_i,i)=0 
            $$
            So let's assume $\delta > 0$. First we note that for every $\varepsilon>0$ such that $z_i - \varepsilon > 0$ we must have $\tau^*(z_i-\varepsilon,i)\geq i$, because otherwise $z_i$ would not be minimal in the solution set of $H_i(\cdot)=0.$ \\ Moreover it holds that $\phi^{\tau^*(z_i-\varepsilon,i)}(z_i,i) \leq \phi^{\overline{\tau}_i}(z_i,i)$ by definition of $\overline{\tau}_i$.
            Using the Lipschitz property of $\Phi(\cdot,i)$ and $\phi^{\tau^*(z_i - \varepsilon,i)}(\cdot,i)$, we hence get the following contradiction:
            \begin{align*}
                0 < \delta \leq \Phi(z_i,i) - \phi^{\tau^*(z_i-\varepsilon,i)}(z_i,i) \pm \Phi(z_i-\varepsilon,i) \leq 2 \varepsilon \to 0 \text{ for } \varepsilon \to 0.
                \\[-3em]
            \end{align*}
            \qed

\end{document}